\documentclass[10pt]{book}
\usepackage[sectionbib]{natbib}
\usepackage{array,epsfig,fancyheadings,rotating}
\usepackage[]{hyperref}

\usepackage{amsmath}
\usepackage{amssymb}
\usepackage{amsfonts}
\usepackage{multirow}
\usepackage{amsthm}
\usepackage{algpseudocode,algorithm}
\usepackage{enumerate}
\usepackage{caption, subcaption}

\newtheorem{theorem}{Theorem}

\theoremstyle{definition}

\begin{document}


\renewcommand{\baselinestretch}{2}

\markright{ \hbox{\footnotesize\rm Statistica Sinica
}\hfill\\[-13pt]
\hbox{\footnotesize\rm
}\hfill }

\markboth{\hfill{\footnotesize\rm Bai Jiang, Tung-yu Wu, Charles Zheng and Wing Wong} \hfill}
{\hfill {\footnotesize\rm Learning Summary Statistic for ABC via DNN} \hfill}

\renewcommand{\thefootnote}{}
$\ $\par


\fontsize{10.95}{14pt plus.8pt minus .6pt}\selectfont
\vspace{0.8pc}
\centerline{\large\bf Learning Summary Statistic for Approximate Bayesian}
\vspace{2pt}
\centerline{\large\bf Computation via Deep Neural Network}
\vspace{.4cm}
\centerline{Bai Jiang, Tung-Yu Wu, Charles Zheng and Wing H. Wong}
\vspace{.4cm}
\centerline{\it Stanford University}
\vspace{.55cm}
\fontsize{9}{11.5pt plus.8pt minus .6pt}\selectfont


\begin{quotation}
\noindent {\it Abstract:}
Approximate Bayesian Computation (ABC) methods are used to approximate posterior distributions in models with unknown or computationally intractable likelihoods. Both the accuracy and computational efficiency of ABC depend on the choice of summary statistic, but outside of special cases where the optimal summary statistics are known, it is unclear which guiding principles can be used to construct effective summary statistics. In this paper we explore the possibility of automating the process of constructing summary statistics by training deep neural networks to predict the parameters from artificially generated data: the resulting summary statistics are approximately posterior means of the parameters. With minimal model-specific tuning, our method constructs summary statistics for the Ising model and the moving-average model, which match or exceed theoretically-motivated summary statistics in terms of the accuracies of the resulting posteriors.\par

\vspace{9pt}
\par
\end{quotation}\par

\def\thefigure{\arabic{figure}}
\def\thetable{\arabic{table}}

\fontsize{10.95}{14pt plus.8pt minus .6pt}\selectfont

\setcounter{chapter}{1}
\setcounter{equation}{0} 
\noindent {\bf 1. Introduction}
\par \smallskip
\noindent {\bf 1.1. Approximate Bayesian Computation}
\par \smallskip
Bayesian inference is traditionally centered around the ability to compute or sample from the posterior distribution of the parameters, having conditioned on the observed data. Suppose data $X$ is generated from a model $\mathcal{M}$ with parameter $\theta$, the prior of which is denoted by $\pi(\theta)$. If the closed form of the likelihood function $l(\theta) = p(X|\theta)$ is available, the posterior distribution of $\theta$ given observed data $x_{obs}$ can be computed via Bayes' rule
$$\pi(\theta|x_{obs}) = \frac{\pi(\theta)p(x_{obs}|\theta)}{p(x_{obs})}.$$
Alternatively, if the likelihood function can only be computed conditionally or up to a normalizing constant, one can still draw samples from the posterior by using stochastic simulation techniques such as Markov Chain Monte Carlo (MCMC) and rejection sampling (\cite{asmussen2007stochastic}).

\lhead[\footnotesize\thepage\fancyplain{}\leftmark]{}\rhead[]{\fancyplain{}\rightmark\footnotesize\thepage}

In many applications, the likelihood function $l(\theta) = p(X|\theta)$ cannot be explicitly obtained, or is intractable to compute; this precludes the possibility of direct computation or MCMC sampling. In these cases, approximate inference can still be performed as long as 1) it is possible to draw $\theta$ from the prior $\pi(\theta)$, and 2) it is possible to simulate $X$ from the model $\mathcal{M}$ given $\theta$, using the methods of Approximate Bayesian Computation (ABC) (See e.g. \cite{beaumont2002approximate, toni2009approximate, lopes2010abc, beaumont2010approximate, csillery2010approximate, marin2012approximate, sunnaaker2013approximate}).

While many variations of the core approach exist, the fundamental idea underlying ABC is quite simple: that one can use rejection sampling to obtain draws from the posterior distribution $\pi(\theta|x_{obs})$ without computing any likelihoods.  We draw parameter-data pairs $(\theta', X')$ from the prior $\pi(\theta)$ and the model $\mathcal{M}$, and accept only the $\theta'$ such that $X' = x_{obs}$, which occurs with conditional probability $P(X = x_{obs}|\theta')$ for any $\theta'$.  Algorithm \ref{alg: ABC1} describes the ABC method for discrete data  (\cite{tavare1997inferring}), which yields an i.i.d. sample $\{\theta^{(i)}\}_{1 \leq i \leq n}$ of the exact posterior distribution $\pi(\theta|X=x_{obs})$.

\begin{algorithm}[h!]
\caption{ABC rejection sampling 1}
\label{alg: ABC1}
\begin{algorithmic}
\For{$i=1,...,n$}
	\Repeat
		\State Propose $\theta' \sim \pi(\theta)$
		\State Draw $X' \sim \mathcal{M}$ given $\theta'$
	\Until{$X' = x_{obs}$ (\textbf{acceptance criterion})}
	\State Accept $\theta'$ and let $\theta^{(i)} = \theta'$
\EndFor
\end{algorithmic}
\end{algorithm}

The success of Algorithm \ref{alg: ABC1} depends on acceptance rate of proposed parameter $\theta'$. For continuous $x_{obs}$ and $X'$, the event $X'=x_{obs}$ happens with probability $0$, and hence Algorithm \ref{alg: ABC1} is unable to produce any draws.  As a remedy, one can relax the acceptance criterion $X'=x_{obs}$ to be $\Vert X'-x_{obs}\Vert < \epsilon$, where $\Vert \cdot \Vert$ is a norm and $\epsilon$ is the tolerance threshold. The choice of $\epsilon$ is crucial for balancing efficiency and approximation error, since with smaller $\epsilon$ the approximation error decreases while the acceptance probability also decreases.
\par \bigskip

\noindent {\bf 1.2. Summary Statistic}
\par \smallskip
When data vectors $x_{obs},X$ are high-dimensional, the inefficiency of rejection sampling in high dimensions results in either extreme inaccuracy, or accuracy at the expense of an extremely time-consuming procedure. To circumvent the problem, one can introduce low-dimensional summary statistic $S$ and further relax the acceptance criterion to be $\Vert S(X')-S(x_{obs})\Vert < \epsilon$. The use of summary statistics results in Algorithm \ref{alg: ABC2}, which was first proposed as the extension of Algorithm \ref{alg: ABC1} in population genetics application (\cite{fu1997estimating, weiss1998inference, pritchard1999population}).

\begin{algorithm}[h!]
\caption{ABC rejection sampling 2}
\label{alg: ABC2}
\begin{algorithmic}
\For{$i=1,...,n$}
	\Repeat
		\State Propose $\theta' \sim \pi$
		\State Draw $X' \sim \mathcal{M}$ with $\theta'$
	\Until{$\Vert S(X')-S(x_{obs})\Vert < \epsilon$ (\textbf{relaxed acceptance criterion})}
	\State Accept $\theta'$ and let $\theta^{(i)} = \theta'$
\EndFor
\end{algorithmic}
\end{algorithm}

Instead of the exact posterior distribution, the resulting sample $\{\theta^{(i)}\}_{1\leq i \leq n}$ obtained by Algorithm \ref{alg: ABC2} follows an approximate posterior distribution
\begin{align}
\pi(\theta|\Vert S(X')-S(x_{obs})\Vert < \epsilon) & \approx \pi(\theta|S(X)=S(x_{obs})) \label{eqn: error due to epsilon}\\
& \approx \pi(\theta|X=x_{obs}). \label{eqn: error due to S}
\end{align}
The choice of the summary statistic is crucial for the approximation quality of ABC posterior distribution. An effective summary statistic should offer a good trade-off between two approximation errors (\cite{blum2013comparative}). The approximation error (\ref{eqn: error due to epsilon}) is introduced when one replaces ``equal" with ``similar" in the first relaxation of the acceptance criterion. Under appropriate regularity conditions, it vanishes as $\epsilon \to 0$. The approximation error (\ref{eqn: error due to S}) is introduced when one compares summary statistics $S(X)$ and $S(x_{obs})$ rather than the original data $X$ and $x_{obs}$. In essence, this is just the information loss of mapping high-dimensional $X$ to low-dimensional $S(X)$. A summary statistic $S$ of higher dimension is in general more informative, hence reduces the approximation error (\ref{eqn: error due to S}). At the same time, increasing the dimension of the summary statistic slows down the rate that the approximation error (\ref{eqn: error due to epsilon}) vanishes in the limit of $\epsilon \to 0$. Ideally, we seek a statistic which is simultaneously low-dimensional and informative.

A sufficient statistic is an attractive option, since sufficiency, by definition, implies that the approximation error (\ref{eqn: error due to S}) is zero (\cite{kolmogorov1942definition, lehmann2006theory}).  However, the sufficient statistic has generally the same dimensionality as the sample size, except in special cases such as exponential families. And even when a low-dimensional sufficient statistic exists, it may be intractable to compute.

The main task of this article is to construct low-dimensional and informative summary statistics for ABC methods. Since our goal is to compare methods of constructing summary statistics (rather than present a complete methodology for ABC), the relatively simple Algorithm \ref{alg: ABC2} suffices. In future work, we plan to use our approach for constructing summary statistics alongside more sophisticated variants of ABC methods, such as those which combine ABC with Markov chain Monte Carlo or sequential techniques (\cite{marjoram2003markov, sisson2007sequential}). Hereafter all ABC procedures mentioned use Algorithm \ref{alg: ABC2}.
\par \bigskip

\noindent {\bf 1.3. Related Work and Our DNN Approach}
\par \smallskip
Existing methods for constructing summary statistics can be roughly classified into two classes, both of which require a set of candidate summary statistics $\mathcal{S}_c = \{S_{c,k}\}_{1 \le k \le K}$ as input. The first class consists of approaches for \textit{best subset selection}. Subsets of $\mathcal{S}_c$ are evaluated according to various information-based criteria, e.g. measure of sufficiency (\cite{joyce2008approximately}), entropy (\cite{nunes2010optimal}), Akaike and Bayesian information criteria (\cite{blum2013comparative}), and the ``best" subset is chosen to be the summary statistic. The second class is \textit{linear regression} approach, which constructs summary statistics by linear regression of response $\theta$ on candidate summary statistics $\mathcal{S}_c$ (\cite{wegmann2009efficient, fearnhead2012constructing}). Regularization techniques have also been considered to reduce overfitting in the regression models (\cite{blum2013comparative}). Many of these methods rely on expert knowledge to provide candidate summary statistics.

In this paper, we propose to automatically learn summary statistics for high-dimensional $X$ by using deep neural networks (DNN). Here DNN is expected to effectively learn a good approximation to the posterior mean $\mathbb{E}_\pi [\theta|X]$ when constructing a minimum squared error estimator $\hat{\theta}(X)$ on a large data set $\{(\theta^{(i)},X^{(i)})\}_{1 \le i \le N} \sim \pi \times \mathcal{M}$. The minimization problem is given by
$$\min_{\beta} \frac{1}{N} \sum_{i=1}^N \left\Vert f_\beta(X^{(i)}) - \theta^{(i)} \right\Vert_2^2,$$
where $f_\beta$ denotes a DNN with parameter $\beta$. The resulting estimator $\hat{\theta}(X) = f_{\hat{\beta}}(X)$ approximates $\mathbb{E}_\pi [\theta|X]$ and further serves as the summary statistic for ABC.

Our motivation for using (an approximation to) $\mathbb{E}_\pi [\theta|X]$ as a summary statistic for ABC is inspired by the semi-automatic method in (\cite{fearnhead2012constructing}). Their idea is that $\mathbb{E}_\pi[\theta|X]$ as summary statistic leads to an ABC posterior, which has the same mean as the exact posterior in the limit of $\epsilon \to 0$. Therefore they proposed to linearly regress $\theta$ on candidate summary statistics $\mathcal{S}_c = \{S_{c,k}\}_{1 \le k \le K}$
$$\min_{\beta} \frac{1}{N} \sum_{i=1}^N \left\Vert \beta_0 + \sum_{k=1}^K \beta_k S_{c,k}(X^{(i)}) - \theta^{(i)} \right\Vert_2^2,$$
and to use the resulting minimum squared error estimator $\hat{\theta}(X)$ as the summary statistic for ABC. In their semi-automatic method, $\mathcal{S}_c$ could be expert-designed statistics or polynomial bases (e.g. power terms of each component $X_j$).

Our DNN approach aims to achieve a more accurate approximation $\hat{\theta}(X) \approx \mathbb{E}_\pi[\theta|X]$ and a higher degree of automation in constructing summary statistics than the semi-automatic method. First, DNN with multiple hidden layers offers stronger representational power, compared to the semi-automatic method using linear regression. A DNN is expected to better approximate $\mathbb{E}_\pi[\theta|X]$ if the posterior mean is a highly non-linear function of $X$. Second, DNNs simply use the original data vector $X$ as the input, and automatically learn the appropriate nonlinear transformations as summaries from the raw data, in contrast to the semi-automatic method and many other existing methods requiring a set of expert-designed candidate summary statistics or a basis expansion. Therefore our approach achieves a higher degree of automation in constructing summary statistics.

\cite{blum2010non} have considered fitting a feed-forward neural network (FFNN) with single hidden layer by regressing $\theta^{(i)}$ on $X^{(i)}$. Their method significantly differs from ours, as theirs was originally motivated by reducing the error between the ABC posterior and the true posterior, rather than constructing summary statistics. Specifically, their method assumes that the appropriate summary statistic $S$ has already been given, and adjusts each draw $(\theta,X)$ from the ABC procedure using summary statistic $S$ in the way
$$\theta^* = m(S(x_{obs})) + \left[\theta - m(S(X))\right] \times \frac{\sigma(S(x_{obs}))}{\sigma(S(X))}.$$
Both $m(\cdot)$ and $\sigma(\cdot)$ are non-linear functions represented by FFNNs. Another key difference is the network size: the FFNNs in \cite{blum2010non} contained four hidden neurons in order to reduce dimensionality of summary statistics, while our DNN approach contains hundreds of hidden neurons in order to gain representational power.
\par \bigskip

\noindent {\bf 1.4. Organization}
\par \smallskip
The rest of the article is organized as follows. In Section 2,  we show how to approximate the posterior mean $\mathbb{E}_\pi[\theta|X]$ by training DNNs. In Sections 3 and 4, we report simulation studies on the Ising model and the moving average model of order 2, respectively. We describe in the supplementary materials the implementation details of DNNs and how consistency can be obtained by using the posterior mean of a basis of functions of the parameters.

\par \bigskip

\setcounter{chapter}{2}
\setcounter{equation}{0} 
\noindent {\bf 2. Methods}
\par \smallskip
Throughout the paper, we denote by $X \in \mathbb{R}^p$ the data, and by $\theta \in \mathbb{R}^q$ the parameter. We assume it is possible to obtain a large number of independent draws $X$ from the model $\mathcal{M}$ given $\theta$ despite the unavailability of $p(X|\theta)$. Denote by $x_{obs}$ the observed data, $\pi$ the prior of $\theta$, $S$ the summary statistic, $\Vert \cdot \Vert$ the norm to measure $S(X)-S(x_{obs})$, and $\epsilon$ the tolerance threshold. Let $\pi_{ABC}^\epsilon(\theta) = \pi(\theta| \Vert S(X)-S(x_{obs})\Vert < \epsilon)$ denote the approximate posterior distribution obtained by Algorithm \ref{alg: ABC2}.

The main task is to construct a low-dimensional and informative summary statistic $S$ for high-dimensional $X$, which will enable accurate approximation of $\pi_{ABC}^\epsilon$. We are interested mainly in the regime where ABC is most effective: settings in which the dimension of $X$ is moderately high (e.g. $p=100$) and the dimension of $\theta$ is low (e.g. $q=1,2,3$). Given a prior $\pi$ for $\theta$, our approach is as follows.
\begin{enumerate}[(1)]
\item Generate a data set $\left\{(\theta^{(i)},X^{(i)})\right\}_{1 \leq i \leq N}$ by repeatedly drawing $\theta^{(i)}$ from $\pi$ and drawing $X^{(i)}$ from $\mathcal{M}$ with $\theta^{(i)}$.
\item Train a DNN with $\{X^{(i)}\}_{1\leq i \leq N}$ as input and $\{\theta^{(i)}\}_{1 \leq i \leq N}$ as target.
\item Run ABC Algorithm \ref{alg: ABC2} with prior $\pi$ and the DNN estimator $\hat{\theta}(X)$ as summary statistic.
\end{enumerate}
Our motivation for training such a DNN is that the resulting statistic (estimator) should approximate the posterior mean $S(X) = \hat{\theta}(X) \approx \mathbb{E}_\pi[\theta|X]$.
\par \bigskip

\noindent {\bf 2.1 Posterior Mean as Summary Statistic}
\par \smallskip
The main advantage of using the posterior mean $\mathbb{E}_\pi[\theta|X]$ as a summary statistic is that the ABC posterior $\pi_{ABC}^\epsilon(\theta) = \pi(\theta|\Vert S(X) - S(x_{obs})\Vert < \epsilon)$ will then have the same mean as the exact posterior in the limit of $\epsilon \to 0$. That is to say, $\mathbb{E}_\pi[\theta|X]$ does not lose any first-order information when summarizing $X$.

This theoretical result has been discussed in Theorem 3 in \cite{fearnhead2012constructing}, but their proof is not rigorous. We provide in Theorem \ref{Theorem1} a more rigorous and general proof.

\begin{theorem} \label{Theorem1}
If $\mathbb{E}_\pi \left[|\theta|\right] < \infty$, then $S(x) = \mathbb{E}_\pi[\theta|X=x]$ is well defined. The ABC procedure with observed data $x_{obs}$, summary statistics $S$, norm $\Vert \cdot \Vert$, and tolerance threshold $\epsilon$ produces a posterior distribution
$$\pi_{ABC}^\epsilon(\theta) = \pi(\theta|\Vert S(X) - S(x_{obs})\Vert < \epsilon),$$
with
$$\Vert \mathbb{E}_{\pi_{ABC}^\epsilon} [\theta] - S(x_{obs}) \Vert < \epsilon,$$
$$\lim_{\epsilon \to 0}  \mathbb{E}_{\pi_{ABC}^\epsilon} [\theta] = \mathbb{E}_\pi[\theta|X=x_{obs}].$$
\end{theorem}

\begin{proof}
First, we show $S(X) = \mathbb{E}_\pi [\theta | X]$ is a version of conditional expectation of $\theta$ given $S(X)$. Denote by $\sigma(X), \sigma(S(X))$ the $\sigma$-algebras of $X$ and $S(X)$, respectively. $S(X)$ is clearly measurable with respect to $\sigma(X)$, thus $\sigma(S(X)) \subseteq \sigma(X)$. Then
\begin{align}
S(X) &= \mathbb{E}_\pi [S(X)| S(X)] \tag*{[$S(X)$ is known in $\sigma(S(X))$]}\\
&=  \mathbb{E}_\pi [ \mathbb{E}_\pi[\theta|X] | S(X)] \tag*{[Definition of $S$]}\\
&= \mathbb{E}_\pi [\theta |S(X)] \tag*{[Tower property, $\sigma(S(X)) \subseteq \sigma(X)$]}
\end{align}
As $A = \{ \Vert S(X) - S(x_{obs}) \Vert < \epsilon \} \in \sigma(S(X))$, we have by the definition of conditional expectation
$$\mathbb{E}_\pi [ \theta \mathbb{I}_A] = \mathbb{E}_\pi [ S(X) \mathbb{I}_A].$$
It follows that
$$\mathbb{E}_{\pi_{ABC}^\epsilon} [\theta] = \mathbb{E}_\pi \left[ \theta | A \right] = \mathbb{E}_\pi \left[ S(X) | A \right],$$
implying by Jensen's inequality that
\begin{align*}
\Vert \mathbb{E}_{\pi_{ABC}^\epsilon} [\theta] - S(x_{obs})\Vert
&= \Vert \mathbb{E}_\pi \left[ S(X) | A \right] - S(x_{obs})\Vert\\
&\leq \mathbb{E}_\pi \left[ \left. \Vert S(X) - S(x_{obs})\Vert \right| A\right]\\
&< \epsilon
\end{align*}
Letting $\epsilon \to 0$ yields $\mathbb{E}_{\pi_{ABC}^\epsilon}[\theta] \to S(x_{obs})= \mathbb{E}_\pi\left[\theta|X=x_{obs}\right]$.
\end{proof}

ABC procedures often give the sample mean of the ABC posterior as the point estimate for $\theta$. Theorem \ref{Theorem1} shows ABC procedure using $\mathbb{E}_\pi[\theta|X]$ as the summary statistic maximizes the point-estimation accuracy in the sense that the exact mean of ABC posterior $\mathbb{E}_{\pi_{ABC}^\epsilon} [\theta]$ is an $\epsilon$-approximation to the Bayes estimator $\mathbb{E}_\pi[\theta|X=x_\text{obs}]$ under squared error loss.

Users of Bayesian inference generally desire more than just point estimates: ideally, one approximates the posterior $\pi(\theta|x_{obs})$ globally. We observe that such a global approximation result is possible when extending Theorem \ref{Theorem1}: if one considers a basis of functions on the parameters, $b(\theta) = (b_1(\theta), ... , b_K(\theta))$, and uses the $K$-dimensional statistic $\mathbb{E}_\pi [ b(\theta)|X]$ as the summary statistic(s), the ABC posterior weakly converges to the exact posterior as $\epsilon \to 0$ and $K \to \infty$ at the appropriate rate. We state this result in the supplementary material.

There is a nice connection between the posterior mean and the sufficient statistics, especially minimal sufficient statistics in the exponential family. If there exists a sufficient statistic $S^*$ for $\theta$, then from the concept of the sufficiency in the Bayesian context (\cite{kolmogorov1942definition}) it follows that for almost every $x$, $\pi(\theta|X=x) = \pi(\theta|S^*(X)=S^*(x))$, and further $S(x) = \mathbb{E}_\pi[\theta|X=x] = \mathbb{E}_\pi[\theta|S^*(X)=S^*(x)]$ is a function of $S^*(x)$. In the special case of an exponential family with minimal sufficient statistic $S^*$ and parameter $\theta$, the posterior mean $S(X) = \mathbb{E}_\pi[\theta|X]$ is a one-to-one function of $S^*(X)$, and thus is a minimal sufficient statistic.
\par \bigskip

\noindent {\bf 2.2. Structure of Deep Neural Network}
\par \smallskip
At a high level, a deep neural network merely represents a non-linear function for transforming input vector $X$ into output $\hat{\theta}(X)$.  The structure of a neural network can be described as a series of $L$ nonlinear transformations applied to $X$. Each of these $L$ transformations is described as a \emph{layer}: where the original input is $X$, the output of the first transformation is the 1st layer, the output of the second transformation is the 2nd layer, and so on, with the output as the $(L+1)$th layer.  The layers 1 to $L$ are called \emph{hidden layers} because they represent intermediate computations, and we let $H^{(l)}$ denote the $l$-th hidden layer. Then the explicit form of the network is
\begin{align*}
H^{(1)} &= \tanh\left(W^{(0)}H^{(0)}+b^{(0)}\right),\\
H^{(2)} &= \tanh\left(W^{(1)}H^{(1)}+b^{(1)}\right),\\
...\\
H^{(L)} &= \tanh\left(W^{(L-1)}H^{(L-1)}+b^{(L-1)}\right),\\
\hat{\theta} &= W^{(L)} H^{(L)} + b^{(L)}.
\end{align*}
where $H^{(0)}=X$ is the input, $\hat{\theta}$ is the output, $W^{(l)}$ and $b^{(l)}$ are the parameters controlling how the inputs of layer $l$ are transformed into the outputs of layer $l$. Let $n^{(l)}$ denote the size of the $l$-th layer: then $W^{(l)}$ is an $n^{(l+1)} \times n^{(l)}$ matrix, called the \emph{weight matrix}, and $b^{(l)}$ is an $n^{(l+1)}$-dimensional vector, called the \emph{bias vector}. The $n^{(l)}$ components of each layer $H^{(l)}$ are also described evocatively as ``neurons'' or ``hidden units". Figure~\ref{fig: dnn} illustrates an example of 3-layer DNN with input $X \in \mathbb{R}^4$ and 5/5/3 neurons in the 1st/2nd/3rd hidden layer, respective.
\begin{figure}[h!]
\begin{center}
\includegraphics[width=1.0\textwidth]{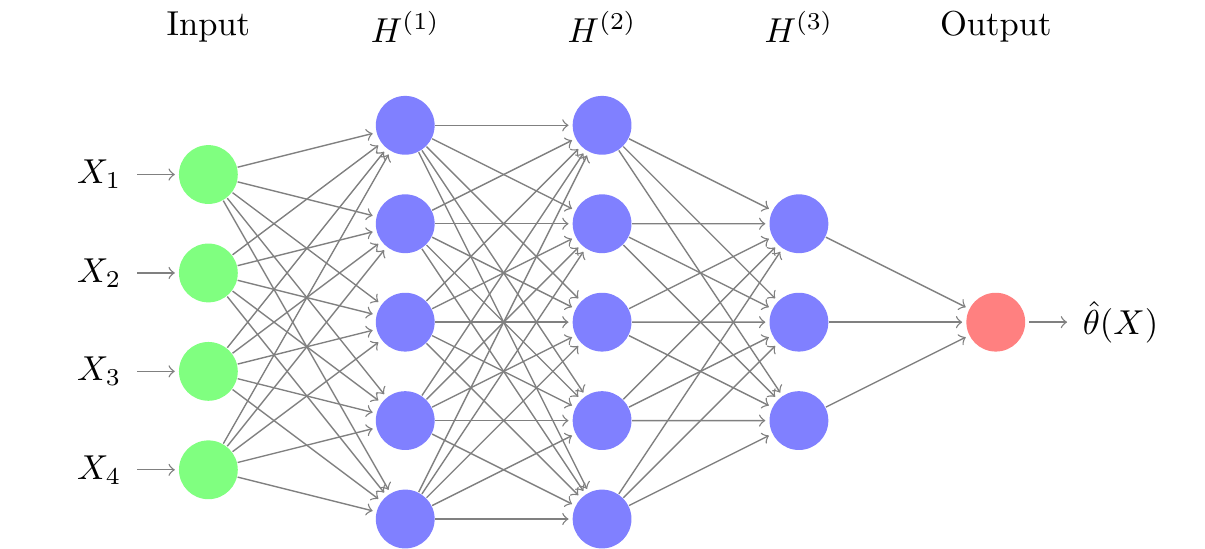}
\end{center}
\caption{An example of DNN with three hidden layers.} \label{fig: dnn}
\end{figure}

\begin{figure}[h!]
\begin{center}
\includegraphics[width=0.9\textwidth]{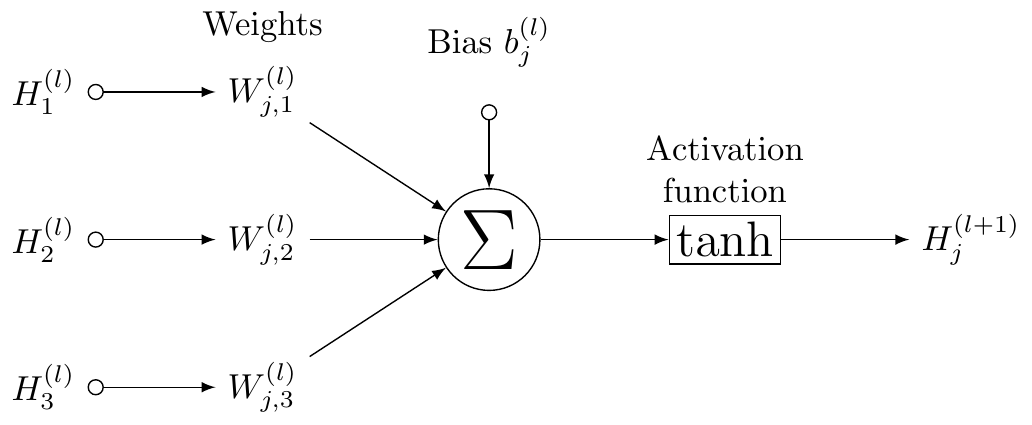}
\end{center}
\caption{Neuron $j$ in the hidden layer $l+1$} \label{fig: neuron}
\end{figure}

The role of layer $l+1$ is to apply a nonlinear transformation to the outputs of layer $l$, $H^{(l)}$, and then output the transformed outputs as $H^{(l+1)}$.  First, a linear transformation is applied to the previous layer $H^{(l)}$, yielding $W^{(l)} H^{(l)} + b^{(l)}$.  The nonlinearity (in this case $\tanh$) is applied to each element of $W^{(l)} H^{(l)} + b^{(l)}$ to yield the output of the current layer, $H^{(l+1)}$. The nonlinearity is traditionally called the ``activation'' function, drawing an analogy to the properties of biological neurons. We choose the function $\tanh$ as an activation function due to smoothness and computational convenience. Other popular choices for activation function are $\text{sigmoid}(t) = \frac{1}{1+\exp (-t)}$ and $\text{ReLU}(t) = \max\{t,0\}$. To better explain the activity of each individual neuron, we illustrate how neuron $j$ in the hidden layer $l+1$ works in Figure~\ref{fig: neuron}.

The output layer takes the top hidden layer $H^{(L)}$ as input and predicts $\hat{\theta} = W^{(L)} H^{(L)} + b^{(L)}$. In many existing applications of deep learning (e.g. computer vision and natural language processing), the goal is to predict a categorical target. In those cases, it is common to use a $\text{softmax}$ transformation in the output layer.  However, since our goal is prediction rather than classification, it suffices to use a linear transformation.
\par \bigskip

\noindent {\bf 2.3. Approximating Posterior Mean by DNN}
\par \smallskip

We use the DNN to construct a summary statistic: a function which maps $x$ to an approximation of $\mathbb{E}_\pi[\theta|X]$.  First, we generate a training set $\mathcal{D}_\pi = \left\{(\theta^{(i)},X^{(i)}), 1 \leq i \leq N \right\}$ by drawing samples from the joint distribution $\pi(\theta, x)$. Next, we train the DNN to minimize the squared error loss between training target $\theta^{(i)}$ and estimation $\hat{\theta}(X^{(i)})$. Thus we minimize (\ref{eqn: loss}) with respect to the DNN parameters $\beta = (W^{(0)}, b^{(0)},...,W^{(L)},b^{(L)})$,
\begin{equation}\label{eqn: loss}
J(\beta) = \frac{1}{N} \sum_{i=1}^N \Vert f_\beta(X^{(i)}) - \theta^{(i)}\Vert_2^2.
\end{equation}
We compute the derivatives using backpropagation (\cite{lecun1998gradient}) and optimize the objective function by stochastic gradient descent method. See the supplementary material for details.

Our approach is based on the fact that any function which minimizes the squared error risk for predicting $\theta$ from $X$ may be viewed as an approximation of the posterior mean $\mathbb{E}_\pi[\theta|X]$. Hence, any supervised learning approach could be used to construct a prediction rule for predicting $\theta$ from $x$, and thereby provide an approximation of $\mathbb{E}_\pi[\theta|X]$. Since in many applications of ABC, we can expect $\mathbb{E}_\pi[\theta|X]$ to be a highly nonlinear and smooth function, it is important to choose a supervised learning approach which has the power to approximate such nonlinear smooth functions.

DNNs appear to be a good choice given their rich representational power for approximating nonlinear functions. More and more practical and theoretical results of deep learning in several areas of machine learning, especially computer vision and natural language processing (\cite{hinton2006reducing, hinton2006fast, bengio2013representation, schmidhuber2015deep}), show that deep architectures composed of simple learning modules in multiple layers can model high-level abstraction in high-dimensional data. It is speculated that by increasing the depth and width of the network, the DNN gains the power to approximate any continuous function; however, rigorous proof of the approximation properties of DNNs remains an important open problem (\cite{farago1993strong, sutskever2008deep, le2010deep}). Nonetheless we expect that DNNs can effectively learn a good approximation to the posterior mean $\mathbb{E}_\pi [\theta|X]$ given a sufficiently large training set.
\par \bigskip

\noindent {\bf 2.4. Avoiding Overfitting}
\par \smallskip
DNN consists of simple learning modules in multiple layers and thus has very rich representational power to learn very complicated relationships between the input $X$ and the output $\theta$. However, DNN is prone to overfitting given limited training data.
In order to avoid overfitting, we consider three methods: generating a large training set, early stopping, and regularization on parameter. 

\textbf{Sufficiently Large Training Data}. This is the fundamental way to avoid overfitting and improve the generalization, that, however, is impossible in many applications of machine learning. Fortunately, in applications of Approximate Bayesian Computation, an arbitrarily large training set can be generated by repeatedly sampling $(\theta^{(i)},X^{(i)})$ from the prior $\pi$ and the model $\mathcal{M}$, and dataset sampling can be parallelized. In our experiments, DNNs contains $3$ hidden layers, each of which has 100 neurons, and has around $3 \times 100 \times 100 = 3 \times 10^4$ parameters, while the training set contains $10^6$ data samples.

\textbf{Early Stopping} (\cite{caruana2001overfitting}). This divides the available data into three subsets: the training set, the validation set and the testing set. The training set is used to compute the gradient and update the parameter. At the same time, we monitor both the training error and the validation error. The validation error usually decreases as does the training error in the early phase of the training process. However, when the network begins to overfit, the validation error begins to increase and we stop the training process. The testing error is reported only for evaluation.

\textbf{Regularization}. This adds an extra term to the loss function that will penalize complexity in neural networks (\cite{nowlan1992simplifying}). Here we consider $\mathcal{L}_2$ regularization (\cite{ng2004feature}) and minimize the objective function
\begin{equation}\label{eqn: loss with penality}
J(\beta; \lambda) = \frac{1}{N} \sum_{i=1}^N \Vert f_\beta(X^{(i)}) - \theta^{(i)}\Vert_2^2 + \lambda \sum_{l=1}^L \Vert W^{(l)}\Vert_\text{F}^2
\end{equation}
where $\Vert W\Vert_\text{F}$ is the Frobenius norm of $W$, the square root of the sum of the absolute squares of its elements.

More sophisticated methods like dropout (\cite{srivastava2014dropout}) and tuning network size can probably better combat overfitting and learn better summary statistic. We only use the simple methods and do minimal model-specific tuning in the simulation studies. Our goal is to show a relatively simple DNN can learn a good summary statistic for ABC.
\par \bigskip

\setcounter{chapter}{3}
\setcounter{equation}{0} 
\noindent {\bf 3. Example: Ising Model}
\par \smallskip
\noindent {\bf 3.1 ABC and Summary Statistics}
\par \smallskip
The Ising model consists of discrete variables ($+1$ or $-1$) arranged in a lattice (Figure~\ref{fig: Iattice}). Each binary variable, called a spin, is allowed to interact with its neighbors.
\begin{figure}[h!]
\begin{center}
\includegraphics[width=0.4\textwidth]{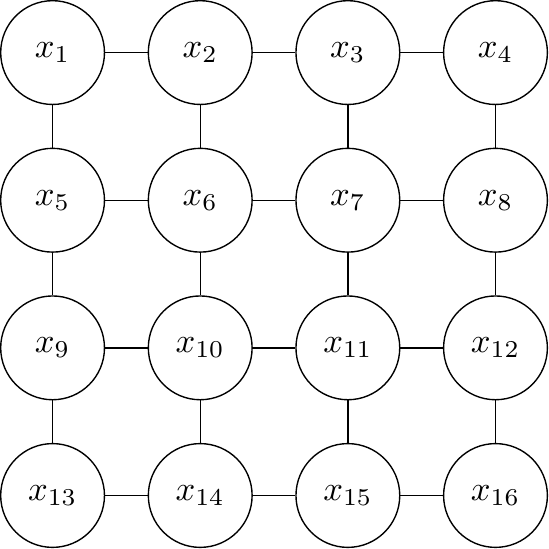}
\end{center}
\caption{Ising model on $4\times 4$ lattice.}
\label{fig: Iattice}
\end{figure}
The inverse-temperature parameter $\theta>0$ characterizes the extent of interaction. Given $\theta$, the probability mass function of the Ising model on $m\times m$ lattice is
$$p(X|\theta) = \frac{ \exp \left(\theta \sum_{j \sim k} X_jX_k \right)}{ Z(\theta)}$$
where $X_j \in \{-1,+1\}$, $j \sim k$ means $X_j$ and $X_k$ are neighbors, and the normalizing constant is
$$Z(\theta) = \sum_{x' \in \{-1,+1\}^{m \times m}} \exp \left(\theta\sum_{j \sim k} x'_jx'_k\right).$$
Since the normalizing constant requires an exponential-time computation, the probability mass function $p(x|\theta)$ is intractable except in small cases.

Despite the unavailability of probability mass function, data $X$ can be still simulated given $\theta$ using Monte Carlo methods such as Metropolis algorithm (\cite{asmussen2007stochastic}). It allows use of ABC for parameter inference. The sufficient statistic $S^*(X) = \sum_{j \sim k} X_jX_k$ is the ideal summary statistic, because $S^*$ is univariate, speeds up the convergence of approximation error (\ref{eqn: error due to epsilon}) in the limit of $\epsilon \to 0$, and losses no information in the approximation (\ref{eqn: error due to S}).

Since $S^*$ results in the ABC posterior with the highest quality, we take it as the gold standard and compare the DNN-based summary statistic to it. The DNN-based summary statistic, if approximating $\mathbb{E}_\pi[\theta|X]$ well, should be an approximately increasing function of $S^*(X)$. As $\mathbb{E}_\pi[\theta|X]$ is an increasing function of $S^*(X)$, it is a sufficient statistic as well. To see this, view the posterior as an exponential family with $S^*(X)$ as ``parameter" and $\theta$ as ``sufficient statistic",
$$\pi(\theta|X) \propto \pi(\theta)p_\theta(X) = \underbrace{\pi(\theta) e^{-\log Z(\theta)}}_{\text{carrier measure}} \exp{\left(S^*(X) \cdot\theta\right)},$$
and then use the mean reparametrization result of exponential family. As $\mathbb{E}_\pi[\theta|X]$ and $S^*(X)$ are highly non-linear functions in the high-dimensional space $\{-1,+1\}^{m \times m}$, they are challenging to approximate.
\par \bigskip

\noindent {\bf 3.2 Experimental Design}
\par \smallskip
\begin{figure}[h!]
\begin{center}
\includegraphics[width=1.0\textwidth]{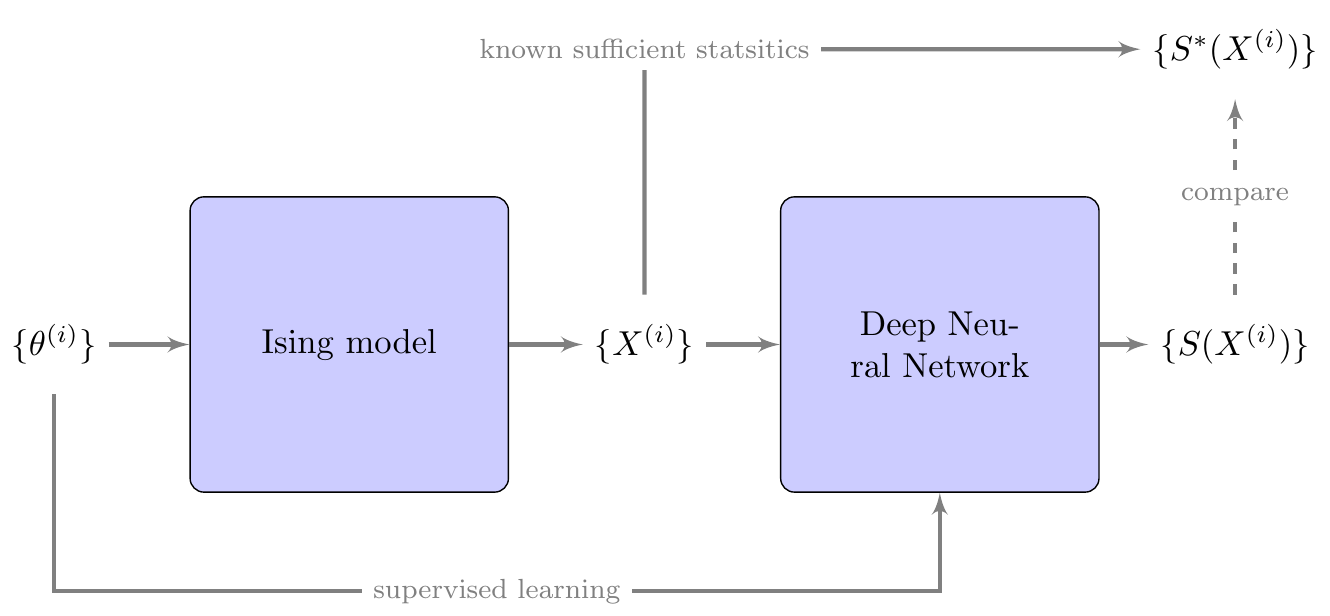}\\
\end{center}
\caption{Experimental design on Ising model}
\label{fig: scheme}
\end{figure}
Figure~\ref{fig: scheme} outlines the whole experimental scheme. We generate a training set by the Metropolis algorithm, train a DNN to learn a summary statistic $S$, and then compare $S(X)$ to the ``gold standard" $S^*(X)$.

Metropolis algorithm generated training, validation, testing sets of size $10^6$,  $10^5$, $10^5$, respectively, from the Ising model on the $10\times10$ lattice with a prior $\pi(\theta) \sim \text{Exp}(\theta_c)$. The value $\theta_c = 0.4406$ is the phase transition point of Ising model on infinite lattice: when $\theta < \theta_c$, the spins tend to be disordered; when $\theta > \theta_c$ is large enough, the spins tend to have the same sign due to the strong neighbor-to-neighbor interactions (\cite{onsager1944crystal}). The Ising model on a finite lattice undergoes a smooth phase transition around $\theta_c$ as $\theta$ increases, which is slightly different than the sharp phase transition on infinite lattice (\cite{landau1976finite}).

A 3-layer DNN with $100$ neurons on each hidden layer was trained to predict $\theta$ from $X$. For the purpose of comparison, the semi-automatic method with components of raw vector $X$ as candidate summary statistics was used. We also tested an FFNN with a single hidden layer of 100 neurons and considered the regularization technique (\ref{eqn: loss with penality}) with $\lambda = 0.001$. The FFNN used is totally different from that used by \cite{blum2010non}. See details in Section 1.3.

Summary statistics learned by different methods led to different ABC posteriors. They were compared to those ABC posteriors resulting from the ideal summary statistic $S^*$.

\par \bigskip

\noindent {\bf 3.3 Results}
\par \smallskip
As shown in Table \ref{tab: Ising prediction}, DNN learns a better prediction rule than the semi-automatic method and FFNN, although it takes more training time. The regularization technique does not improve the performance, probably because overfitting is not a significant issue given that the training data ($N=10^6$) outnumbers the $\approx 3 \times 10^4$ parameters. 

\begin{table}[h!]
\centering
\begin{tabular}{l | c c c}
\hline
Method & Training RMSE & Testing RMSE & Time (s)\\
\hline
Semi-automatic & 0.4401 & 0.4406 & 4.36\\
FFNN, $\lambda=0$  & 0.2541 & 0.2541 & 480.08\\
DNN, $\lambda=0$ & \textbf{0.2319} & \textbf{0.2318} & 1348.17\\
\hline
FFNN, $\lambda=0.001$ & 0.2583 & 0.2584 & 447.07\\
DNN, $\lambda=0.001$ & 0.2514 & 0.2512 & 1378.33\\
\hline
\end{tabular}
\caption{The root-mean-square error (RMSE) and training time of the semi-automatic, FFNN, and DNN methods to predict $\theta$ given $X$. $\lambda$ is the penalty coefficient in the regularized objective function (\ref{eqn: loss with penality}). Stochastic gradient descent fits each FFNN or DNN by 200 full passes (epochs) through the training set.} \label{tab: Ising prediction}
\end{table}

Figure~\ref{fig: Ising scatterplot, deeplearn} displays a scatterplot which compares the DNN-based summary statistic $S$ and the sufficient statistic $S^*$. Points in the scatterplot represent to $(S^*(x),S(x))$ for an instance $x$ in the testing set. A large number of the instances are concentrated at $S^* = 192, 200$, which appear as points in the top-right corner of the scatterplot. These instances are relatively uninteresting, so we display a heatmap of $(S(x), S^*(x))$ excluding them in Figure~\ref{fig: Ising heatmap, deeplearn}. It shows that the DNN-based summary statistic $S(X)$ approximates an increasing function of $S^*(X)$.

\begin{figure}[h!]
    \centering
    \begin{subfigure}[b]{0.42\textwidth}
        \includegraphics[height=2.5in]{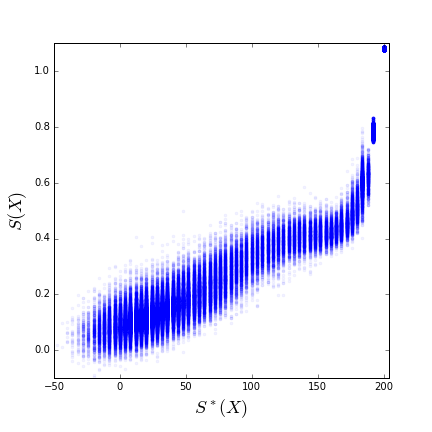}
        \caption{} \label{fig: Ising scatterplot, deeplearn}
    \end{subfigure}
    ~
    \begin{subfigure}[b]{0.53\textwidth}
        \includegraphics[height=2.5in]{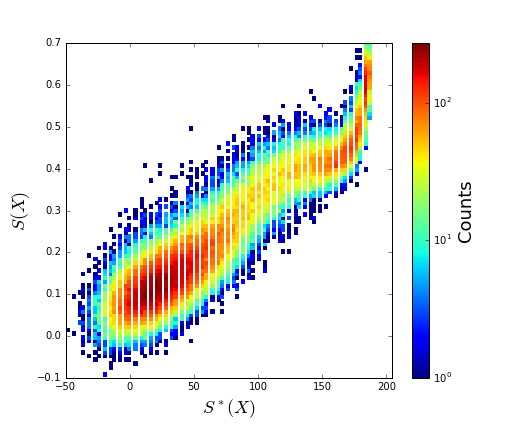}
        \caption{} \label{fig: Ising heatmap, deeplearn}
    \end{subfigure}
    \caption{DNN-based summary statistic $S$ v.s. sufficient statistic $S^*$ on the test dataset. (a) Scatterplot of $10^5$ test instances. Each point represents to $(S^*(x),S(x))$ for a single test instance $x$. (b) Heatmap excluding instances with $S^*(x)=192,200$.} 
\end{figure}

\begin{figure}[h!]
    \centering
    \begin{subfigure}[b]{0.42\textwidth}
        \includegraphics[height=2.5in]{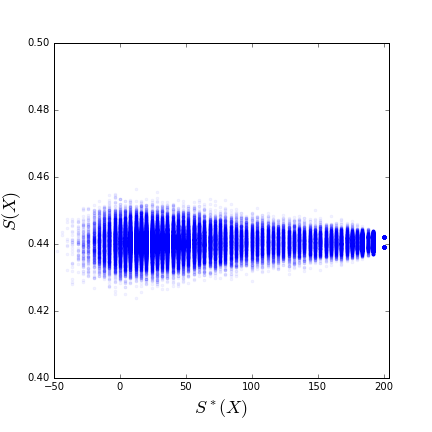}
        \caption{} \label{fig: Ising scatterplot, semiauto}
    \end{subfigure}
    ~
    \begin{subfigure}[b]{0.53\textwidth}
        \includegraphics[height=2.5in]{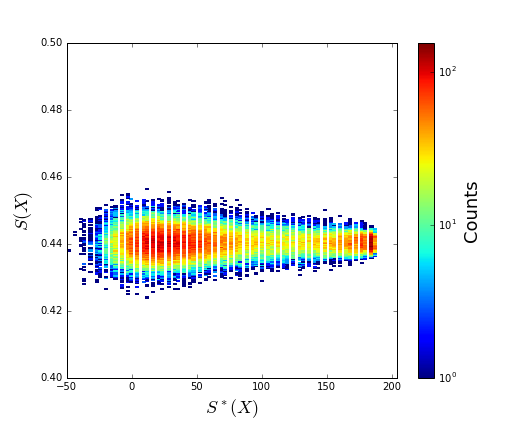}
        \caption{} \label{fig: Ising heatmap, semiauto}
    \end{subfigure}
    \caption{Summary statistic $S$ constructed by the semi-automatic method v.s. sufficient statistic $S^*$ on the test dataset. (a) Scatterplot of $10^5$ test instances. Each point represents to $(S^*(x),S(x))$ for a single test instance $x$. (b) Heatmap excluding instances with $S^*(x)=192,200$.}  \label{fig: Ising prediction, semiauto}
\end{figure}

The semi-automatic method constructs a summary statistic that fails to approximate $\mathbb{E}_\pi[\theta|X]$ (an increasing function of $S^*(X)$) but centers around the prior mean $\theta_c = 0.4406$ (Figure \ref{fig: Ising prediction, semiauto}). This is not surprising since the semi-automatic construction, a linear combination of $X_j$, is unable to capture the non-linearity of $\mathbb{E}_\pi[\theta|X]$.

ABC posterior distributions were obtained with the sufficient statistic $S^*$ and the summary statistics $S$ constructed by DNN and the semi-automatic method. For the sufficient statistic $S^*$, we set the tolerance level $\epsilon = 0$ so that the ABC posterior sample follows the exact posterior $\pi(\theta|X=x_{obs})$. For each summary statistic $S$, we set the tolerance threshold $\epsilon$ small enough so that $0.1\%$ of $10^6$ proposed $\theta'$s were accepted. We repeated the comparison for four different observed data $x_{obs}$, generated from $\theta = 0.2,0.4,0.6,0.8$, respectively; in each case, we compared the posterior obtained from $S^*$ with the posteriors obtained from $S$, in Figure~\ref{fig: Ising posterior}.

\begin{figure}[h!]
\centering
\includegraphics[width=\textwidth]{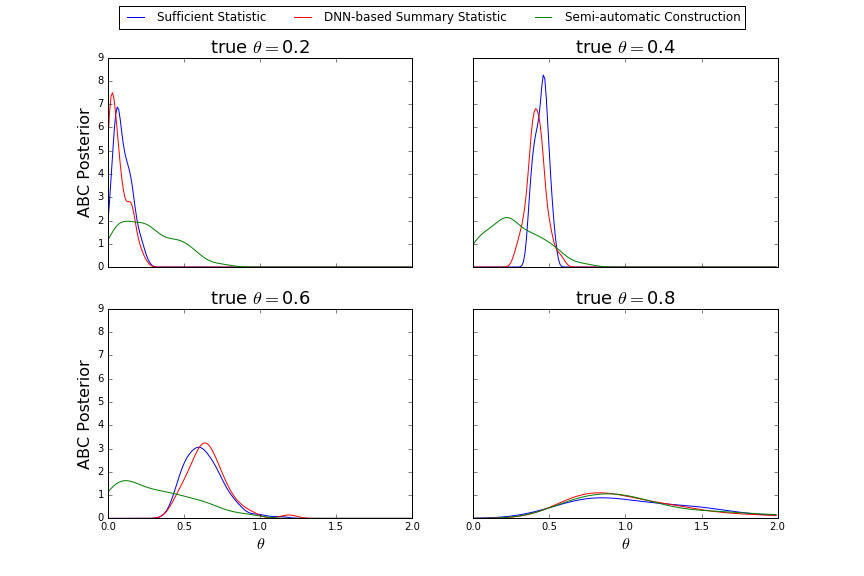}
\caption{ABC posterior distributions for $x_{obs}$ generated with true $\theta = 0.2, 0.4, 0.6, 0.8$.}
\label{fig: Ising posterior}
\end{figure}

We highlights the case with true $\theta=0.8$ (lower-right subplot in Figure~\ref{fig: Ising posterior}). Since with high probability the spins $X_i$ have the same sign when $\theta$ is large, it becomes difficult to distinguish different values of $\theta$ above the critical point $\theta_c$ based on the data $x_{obs}$. Hence we should expect the posterior to be small below $\theta_c$ and have a similar shape to the prior distribution above $\theta_c$. All three ABC posteriors demonstrate this property.
\par \bigskip

\setcounter{chapter}{4}
\setcounter{equation}{0} 
\noindent {\bf 4. Example: Moving Average of Order 2}
\par \smallskip
\noindent {\bf 4.1 ABC and Summary Statistics}
\par \smallskip
The moving-average model is widely used in time series analysis. With $X_1,\hdots, X_p$ the observations, the moving-average model of order $q$, denoted by $\text{MA}(q)$, is given by
\begin{equation*}
X_j = Z_j + \theta_1Z_{j-1} + \theta_2 Z_{j-2} + ... +
\theta_qZ_{j-q}, \quad j=1,...,p,
\end{equation*}
where $Z_j$ are unobserved white noise error terms. We took $Z_j \overset{i.i.d.}{\sim} N(0,1)$ in order to enable exact calculation of the posterior distribution $\pi(\theta|x_{obs})$, and then evaluation of the ABC posterior distribution. If the $Z_j$'s are non-Gaussian, the exact posterior $\pi(\theta|x_{obs})$ is computationally intractable, but ABC is still applicable.

Approximate Bayesian Computation has been applied to study the posterior distribution of the $\text{MA}(2)$ model using the auto-covariance as the summary statistic (\cite{marin2012approximate}). The auto-covariance is a natural choice for the summary statistic in the \text{MA}(2) model because it converges to a one-to-one function of underlying parameter $\theta = (\theta_1, \theta_2)$ in probability as $p \to \infty$ by the Weak Law of Large Numbers,
\begin{align*}
AC_1 &= \frac{1}{p-1}\sum_{j=1}^{p-1} X_{j}X_{j+1} \to \mathbb{E}(X_{1}X_{2}) = \theta_1+\theta_1\theta_2\\
AC_2 &= \frac{1}{p-2}\sum_{j=1}^{p-2} X_{j}X_{j+2} \to \mathbb{E}(X_{1}X_{3}) = \theta_2.
\end{align*}
\par \bigskip

\noindent {\bf 4.2 Experimental Design}
\par \smallskip
The $\text{MA}(2)$ model is identifiable over the triangular region
$$\theta_1 \in [-2,2], \quad \theta_2 \in [-1,1], \quad \theta_2 \pm
\theta_1 \geq -1,$$
so we took a uniform prior $\pi$ over this region, and generated the training, validation, testing sets of size $10^6$,  $10^5$, $10^5$, respectively. Each instance was a time series of length $p=100$. 

A 3-layer DNN with $100$ neurons on each hidden layer was trained to predict $\theta$ from $X$. For purposes of comparison, we constructed the semi-automatic summary statistic by fitting linear regression of $\theta$ on candidate summary statistics - polynomial bases $X_j,X_j^2,X_j^3,X_j^4$. We also test an FFNN with a single hidden layer of 100 neurons and considered the regularization technique (\ref{eqn: loss with penality}) with $\lambda = 0.001$. The FFNN used here is different than that used by \cite{blum2010non}. See details in Section 1.3.

Next, we generated some true parameters $\theta$ from the prior, drew the observed data $x_{obs}$, and numerically computed the exact posterior $\pi(\theta|x_{obs})$. Then we computed ABC posteriors using the auto-covariance statistic $(AC_1,AC_2)$, the DNN-based summary statistics $(S_1,S_2)$, and the semi-automatic summary statistic. The resulting ABC posteriors are compared to the exact posterior and evaluated in terms of the accuracies of the posterior mean of $\theta$, the posterior marginal variances of $\theta_1, \theta_2$, and the posterior correlation between $(\theta_1,\theta_2)$.
\par \bigskip

\noindent {\bf 4.3 Results}
\par \smallskip
\begin{table}[h!]
\centering
\begin{tabular}{l | c c c c c}
\hline
& \multicolumn{2}{c}{Training RMSE} & \multicolumn{2}{c}{Testing RMSE} & Time (s)\\
Method & $\theta_1$ & $\theta_2$ & $\theta_1$ & $\theta_2$ & \\
\hline
Semi-automatic & 0.8150 & 0.3867 & 0.8174 & 0.3857 & 45.63\\
FFNN, $\lambda=0$  & 0.1857 & 0.2091 & 0.1884 & 0.2115 & 543.42\\
DNN, $\lambda=0$ & \textbf{0.1272} & \textbf{0.1355} & \textbf{0.1293} & \textbf{0.1378} & 1402.02\\
\hline
FFNN, $\lambda=0.001$ & 0.2642 & 0.2522 & 0.2679 & 0.2546 & 432.27\\
DNN, $\lambda=0.001$ & 0.1958 & 0.1939 & 0.1980 & 0.1956 & 1282.66\\
\hline
\end{tabular}
\caption{The root-mean-square error (RMSE) and training time of the semi-automatic, FFNN and DNN methods to predict $(\theta_1,\theta_2)$ given $X$. $\lambda$ is the penalty coefficient in the regularized objective function (\ref{eqn: loss with penality}). Stochastic gradient descent fits each FFNN or DNN by 200 full passes (epochs) through the training set.} \label{tab: MA2 prediction}
\end{table}

Again DNN learns a better prediction rule than the semi-automatic method and FFNN, but takes more training time (Table \ref{tab: MA2 prediction}, Figures \ref{fig: MA2 prediction, deeplearn} and \ref{fig: MA2 prediction, semiauto}). The regularization technique does not improve the performance.
\begin{figure}[h!]
    \centering
    \begin{subfigure}[b]{0.45\textwidth}
        \includegraphics[width=2.8in]{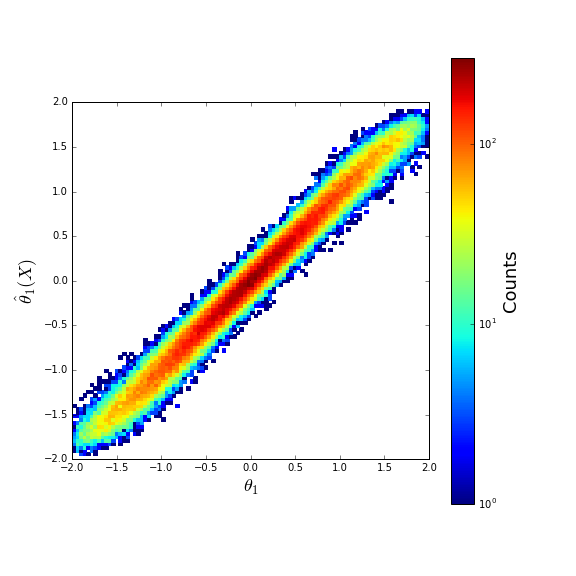}
    \end{subfigure}
    ~
    \begin{subfigure}[b]{0.45\textwidth}
        \includegraphics[width=2.8in]{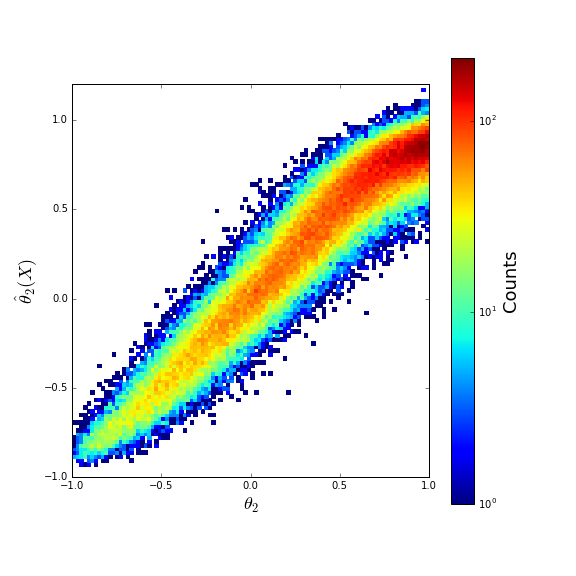}
    \end{subfigure}
    \caption{DNN predicting $\theta_1,\theta_2$ on the test dataset of $10^5$ instances.} \label{fig: MA2 prediction, deeplearn}
\end{figure}

\begin{figure}[h!]
    \centering
    \begin{subfigure}[b]{0.45\textwidth}
        \includegraphics[width=2.8in]{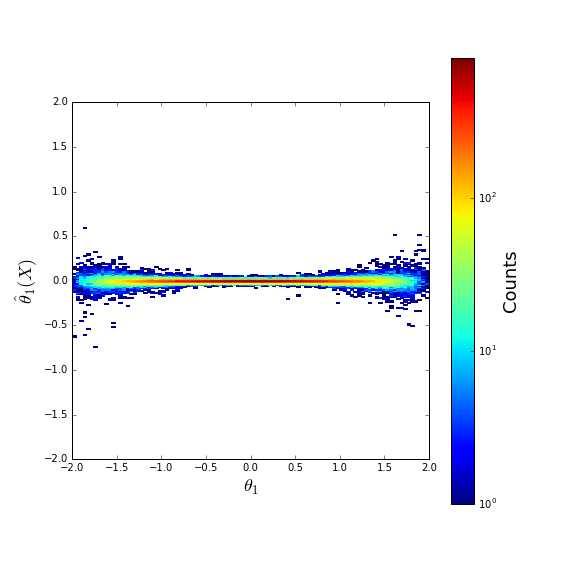}
    \end{subfigure}
    ~
    \begin{subfigure}[b]{0.45\textwidth}
        \includegraphics[width=2.8in]{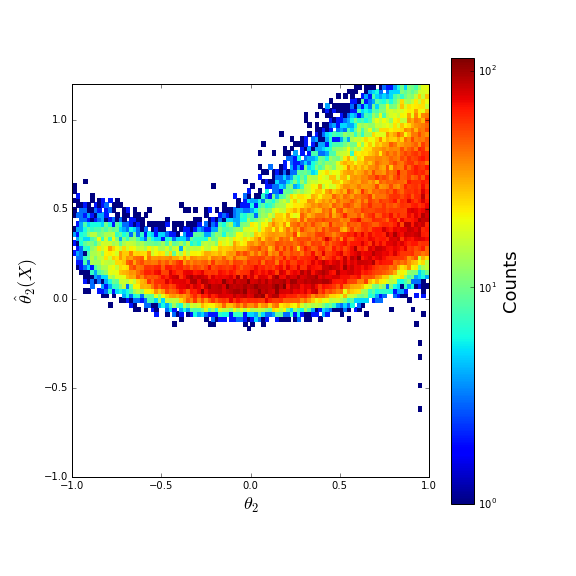}
    \end{subfigure}
    \caption{DNN predicting $\theta_1,\theta_2$ on the test dataset of $10^5$ instances.}  \label{fig: MA2 prediction, semiauto}
\end{figure}

We ran ABC procedures for an observed datum $x_{obs}$ generated by true parameter $\theta = (0.6,0.2)$, with three different choices of summary statistic: the DNN-based summary statistic, the auto-covariance, and also the semi-automatic summary statistic. The tolerance threshold $\epsilon$ was set to accept $0.1\%$ of $10^5$ proposed $\theta'$ in ABC procedures. Figure~\ref{fig: MA2 posterior} compares the ABC posterior draws to the exact posterior which is numerically computed.
\begin{figure}[h!]
\centering
\includegraphics[width=3.78in]{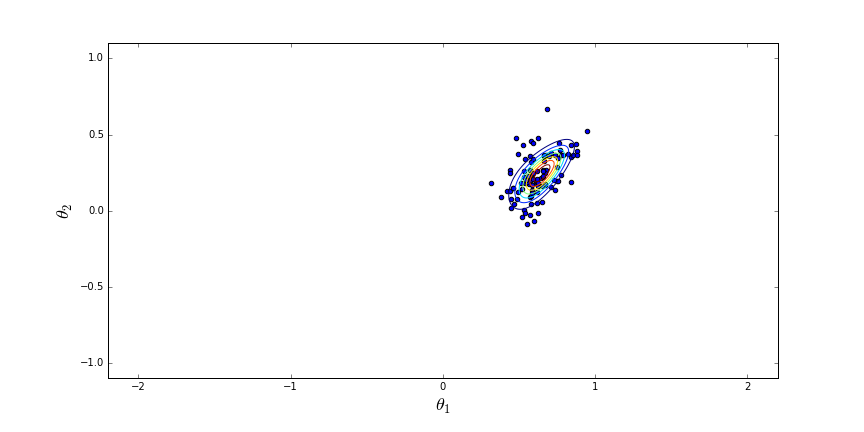}\\
\includegraphics[width=3.78in]{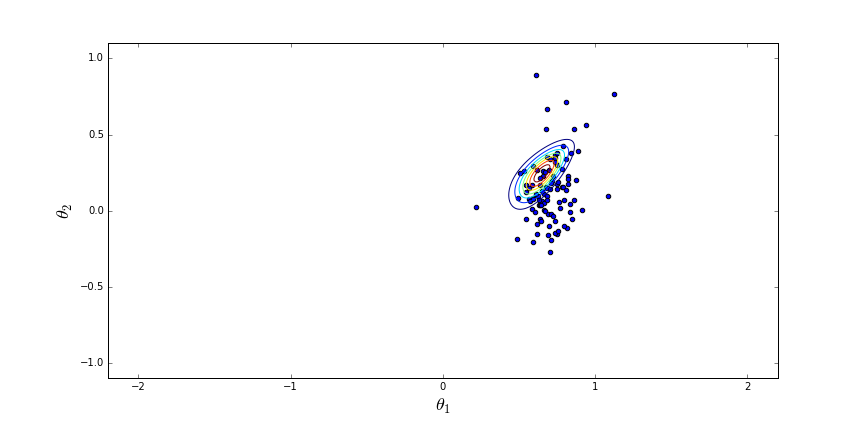}\\
\includegraphics[width=3.78in]{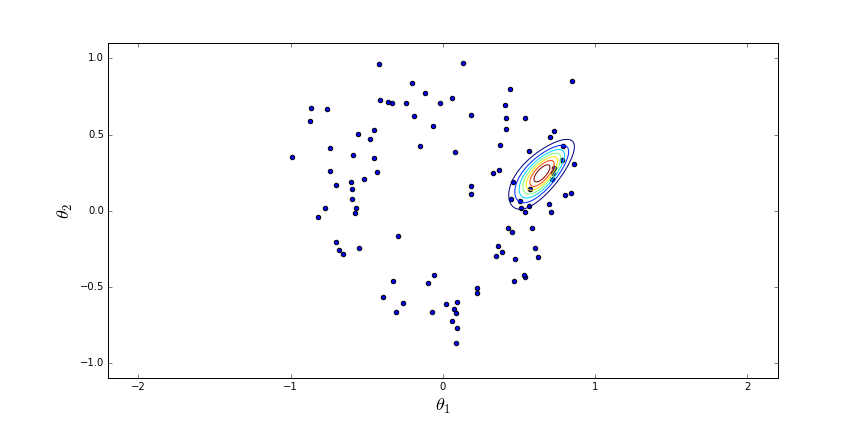}
\caption{ABC posterior draws (top: DNN-based summary statistics, middle: auto-covariance, bottom: semi-automatic construction) for observed data $x_{obs}$ generated with $\theta = (0.6,0.2)$, compared to the exact posterior distribution contours.} 
\label{fig: MA2 posterior}
\end{figure}

The DNN-based summary statistic gives a more accurate ABC posterior than either the ABC posterior obtained by the auto-covariance statistic or the semi-automatic construction. One of the important features of the DNN-based summary statistic is that its ABC posterior correctly captures the correlation between $\theta_1$ and $\theta_2$, while the auto-covariance statistic and the semi-automatic statistic appear to be insensitive to this information (Table~\ref{MA2 abc stats}).
\begin{table}[h!]
\begin{center}
\begin{tabular}{l | c c c c c}
\hline
Posterior & mean($\theta_1$) & mean($\theta_2$) & std($\theta_1$) & std($\theta_2$) & cor($\theta_1,\theta_2$)\\
\hline
Exact  & 0.6418 & 0.2399 & 0.1046 & 0.1100 & 0.6995\\
ABC (DNN) & 0.6230 & 0.2300 & 0.1210 & 0.1410 & 0.4776\\
ABC (auto-cov) & 0.7033 & 0.1402 & 0.1218 & 0.2111 & 0.2606\\
ABC (semi-auto) & 0.0442 &  0.1159 & 0.5160 & 0.4616 & -0.0645\\
\hline
\end{tabular}
\end{center}
\caption{Mean and covariance of exact/ABC posterior distributions for observed data $x_{obs}$ generated with $\theta = (0.6,0.2)$ in Figure \ref{fig: MA2 posterior}.} \label{MA2 abc stats}
\end{table}

We repeated the comparison for 100 different $x_{obs}$. As Table~\ref{MA2 abc stats 100} shows, the ABC procedure with the DNN-based statistic better approximates the posterior moments than those using the auto-covariance statistic and the semi-automatic construction.

\begin{table}[h!]
\begin{center}
\begin{tabular}{l | c c c c c}
\hline
Posterior & \multicolumn{5}{c}{MSE for}\\
& mean($\theta_1$) & mean($\theta_2$) & std($\theta_1$) & std($\theta_2$) & cor($\theta_1,\theta_2$)\\
\hline
ABC (DNN) & 0.0096 & 0.0089 & 0.0025 & 0.0026 & 0.0517\\
ABC (auto-cov) & 0.0111 & 0.0184 & 0.0041 & 0.0065 & 0.1886\\
ABC (semi-auto) & 0.5405 &  0.1440 & 0.4794 & 0.0891 & 0.3116\\
\hline
\end{tabular}
\end{center}
\caption{Mean squared error (MSE) between mean and covariance of exact/ABC posterior distributions for 100 different $x_{obs}$.} \label{MA2 abc stats 100}
\end{table}
\par
\par \bigskip

\setcounter{chapter}{5}
\setcounter{equation}{0} 
\noindent {\bf 5. Discussion}\\
We address how to automatically construct low-dimensional and informative summary statistics for ABC methods, with minimal need of expert knowledge. We base our approach on the desirable properties of the posterior mean as a summary statistic for ABC, though it is generally intractable. We take advantage of the representational power of DNNs to construct an approximation of the posterior mean as a summary statistic.

We only heuristically justify our choice of DNNs to construct the approximation but obtain promising empirical results. The Ising model has a univariate sufficient statistic that is the ideal summary statistic and results in the best achievable ABC posterior. It is a challenging task to construct a summary statistic akin to it due to its high non-linearity and high-dimensionality, but we see in our experiments that the DNN-based summary statistic approximates an increasing function of the sufficient statistic. In the moving-average model of order 2, the DNN-based summary statistic outperforms the semi-automatic construction. The DNN-based summary statistic, which is automatically constructed, outperforms the auto-covariances; the auto-covariances in the MA(2) model can be transformed to yield a consistent estimate of the parameters, and have been widely used in the literature.

A DNN is prone to overfitting given limited training data, but this is not an issue when constructing summary statistics for ABC. In the setting of Approximate Bayesian Computation, arbitrarily many training samples can be generated by repeatedly sampling $(\theta^{(i)},X^{(i)})$ from the prior $\pi$ and the model $\mathcal{M}$. In our experiments, the size of the training data ($10^6$) is much larger than the number of parameters in the neural networks ($10^4$), and there is little discrepancy between the prediction error losses on the training data and the testing data. The regularization technique does not significantly improve the performance.

We compared the DNN with three hidden layers with the FFNN with a single hidden layer. Our experimental comparison indicates that FFNNs are less effective than DNNs for the task of summary statistics construction.

\vskip 14pt
\noindent {\large\bf Supplementary Materials}\\
The supplementary materials contain an extension of Theorem \ref{Theorem1} and show the convergence of the posterior expectation of $b(\theta)$ under the posterior obtained by ABC using $S_b(X) = \mathbb{E}_\pi\left[b(\theta)|X\right]$ as the summary statistic. This extension establishes a global approximation to the posterior distribution.
Implementation details of backpropagation and stochastic gradient descent algorithms when training deep neural network are provided. The derivatives of squared error loss function with respect to network parameters are computed. They are used by stochastic gradient descent algorithms to train deep neural networks.
\par
\vskip 14pt
\noindent {\large\bf Acknowledgements}

The authors gratefully acknowledge the National Science Foundation grants DMS1407557 and DMS1330132.
\par

\markboth{\hfill{\footnotesize\rm Bai Jiang, Tung-yu Wu, Charles Zheng and Wing Wong} \hfill}
{\hfill {\footnotesize\rm Learning Summary Statistic for ABC via DNN} \hfill}

\bibhang=1.7pc
\bibsep=2pt
\fontsize{9}{14pt plus.8pt minus .6pt}\selectfont
\renewcommand\bibname




\end{document}